\documentclass[a4paper,11pt]{article} 

\usepackage[margin=3.5cm]{geometry} 
\usepackage[italian,english]{babel}  
\usepackage[T1]{fontenc}             
\usepackage[utf8]{inputenc}          
\usepackage{microtype}               

\usepackage{csquotes}   
\usepackage{authblk}    

\usepackage{amssymb,amsmath,mathtools,amsthm} 
\usepackage{physics}   
\usepackage{dsfont}    

\usepackage{listings}  

\usepackage{graphicx}   
\graphicspath{{Figures/}} 
\usepackage[section]{placeins} 
\usepackage{subcaption} 

\usepackage{booktabs}   
\usepackage{tabularx}   
\usepackage{adjustbox}  
\usepackage{multirow}   
\usepackage{float}      
\usepackage[font={footnotesize}, hypcap=false]{caption} 
\usepackage[flushleft]{threeparttable} 
\usepackage{makecell}           

\usepackage{tikz} 
\usetikzlibrary{shapes,arrows,backgrounds,positioning,decorations.text}

\tikzstyle{startstop} = [rectangle, rounded corners, minimum width=3cm, minimum height=1cm,text centered, draw=black, fill=gray!20]
\tikzstyle{process} = [rectangle, minimum width=3cm, minimum height=1cm, text centered, draw=black, fill=blue!20]
\tikzstyle{decision} = [diamond, aspect=2, minimum width=3cm, minimum height=1cm, text centered, draw=black, fill=orange!20]
\tikzstyle{arrow} = [thick,->,>=stealth]

\usepackage{natbib} 
\newtheoremstyle{mystyle} 
  {}{}{\itshape}{}{\bfseries}{:}{ }{}
\theoremstyle{mystyle}
\theoremstyle{definition}
\newtheorem{proposition}{Proposition}
\newtheorem{remark}{Remark}

\usepackage[
    colorlinks=true, 
    linkcolor=blue,  
    citecolor=blue,  
    urlcolor=blue,   
    filecolor=blue   
]{hyperref}

\usepackage[nameinlink,noabbrev]{cleveref}
\crefname{hypothesis}{Hypothesis}{Hypotheses}
\Crefname{hypothesis}{Hypothesis}{Hypotheses}
\crefname{result}{Result}{Results}
\Crefname{result}{Result}{Results}

\usepackage{textcomp} 
\usepackage{framed}   
\usepackage[page]{appendix} 

\newcommand\extrafootertext[1]{
    \bgroup
    \renewcommand\thefootnote{\fnsymbol{footnote}}%
    \renewcommand\thempfootnote{\fnsymbol{mpfootnote}}%
    \footnotetext[0]{#1}%
    \egroup
}

\usepackage[draft]{changes} 
\definechangesauthor[name={Marco},  color=orange]{M}
\definechangesauthor[name={Laura},  color=purple]{L}
\definechangesauthor[name={Lorenzo}, color=green]{LS}

\usepackage[most]{tcolorbox}
\tcbuselibrary{breakable,skins}
\usepackage{enumitem} 
\newtcolorbox{algobox}[2][]{%
  enhanced,
  colback=white,
  colframe=black,
  boxrule=0.5pt,
  arc=2mm,
  title={#2},
  fonttitle=\bfseries,
  coltitle=black,
  attach boxed title to top left={xshift=2mm,yshift*=-2mm},
  boxed title style={
    colback=white,
    colframe=black,
    boxrule=0.5pt,
    arc=1.5mm,
    top=1pt, bottom=1pt,
    left=2mm, right=2mm
  },
  #1
}
\setlist[enumerate,1]{label=\textbf{\arabic*.}, leftmargin=*, itemsep=2pt}
\setlist[enumerate,2]{label=(\alph*), leftmargin=*, itemsep=2pt}
\setlist[itemize]{leftmargin=*, itemsep=2pt}

\begin{document}


\title{Nonlinear dynamics of educational choices under social influence and endogenous returns}
\author[a]{Andrea Caravaggio}
\author[b]{Marco Catola}
\author[c,*]{Silvia Leoni}
\affil[a]{\footnotesize Department of Economics and Statistics, University of Siena}
\affil[b]{\footnotesize Department of Law, University of Pisa}
\affil[c]{Department of Economics and Management, University of Florence}
\affil[*]{\footnotesize Corresponding author\vspace{1em}}
\date{}

\maketitle


\begin{abstract} 

Decisions to pursue higher education are not fully explained by economic incentives, with social influence and peer effects playing a crucial, yet dynamically understudied, role. This paper develops a theoretical non-linear dynamics model analysing the interplay between economic returns and social pressure. We model a heterogeneous population of "Followers," who exhibit imitative behaviour, and "Positional Agents," who display counter-adaptive behaviour. Agents' preferences for education evolve endogenously, reacting to both aggregate enrolment and an endogenous wage premium that declines with the supply of educated workers. The aggregate dynamics are governed by a one-dimensional non-linear map. By assuming fixed population structure. we show that the social conflict between pro-cyclical imitative forces and counter-cyclical positional forces can destabilize the steady state, generating a period-doubling route to chaos. These complex, endogenous fluctuations in enrolment emerge only for intermediate, heterogeneous population mixes, while homogeneous populations remain stable. We argue that this instability represents a significant coordination failure, scrambling economic signals and hindering rational long-term planning for both students and institutions, making it a key policy concern. Finally, we also extend the result to the case where the population structure is endogenous.
\noindent

\noindent 
\\[.1cm]
\emph{JEL classification}: $C61$; $D91$; $I21$\\
\emph{Keywords}: Educational choices; Non-linear Dynamics; Peer Effects; Wage Premium; Heterogeneous Agents.
\end{abstract}


\section{Introduction}
In recent years, several studies have highlighted that the decision to enrol in higher education cannot be fully explained by purely economic incentives. In the Italian case, despite relatively low tuition fees and the absence of formal barriers to access, university attainment remains among the lowest in Europe. As shown by \citet{Leoni2022}, this pattern reflects not only the limited wage premium associated with tertiary education (see also \citealt{Viesti2018}; \citealt{Franzini2019}; \citealt{Naticchioni2016}), but also the crucial role of \emph{social influence} in shaping educational choices. In her agent-based model of the Italian higher education system, Leoni models individuals’ preferences for enrolling on three main components: economic motivation, the effort required to complete university studies, and—most importantly—social interaction within one’s reference network. Young people may decide to enrol in university not only because of the expected monetary return, but also because their peers do so; this \emph{educational conformism} generates patterns of collective behaviour that cannot be reduced to individual optimisation. The importance of social and peer effects in educational decisions has been further confirmed by more recent empirical and experimental works: \citet{Pratschke2023} model peer effects in Italian classrooms using a spatial autoregressive framework, showing that endogenous interactions significantly shape students’ trajectories; \citet{Zarate2023} find that peer centrality within social networks amplifies both social and academic skills; and \citet{Carlana2025} highlight that social image concerns and same-gender interactions crucially affect the likelihood of pursuing higher education. Beyond the Italian case, a growing body of research in economics and social dynamics emphasizes how imitation, conformity, and social contagion can account for aggregate patterns in education and inequality (see, among others, \citealt{Dimarco2024}; \citealt{Herstad2024}; \citealt{DeDomenico2025}; \citealt{ChibaOkabe2024}), reinforcing the view that educational choices are socially interdependent and that local interactions can produce complex collective outcomes. The present paper develops a theoretical and dynamic counterpart of this literature: we retain the idea that individual educational choices are shaped by both economic incentives and social influence, but we formalise this mechanism in a stylised analytical way by adopting a \emph{bandwagon--snob} dynamic à la \citet{Caravaggio2020}. In their model of consumption choices, Caravaggio and Sodini—building on the classical insights of \citet{Leibenstein1950}, \citet{Veblen1899}, and \citet{Simmel1904}, and further developed by \citet{Naimzada2009}—show how imitation and distinction among social groups can endogenously generate cyclical and even chaotic collective patterns. 

Transposing this intuition into the context of education, we represent the social influence component of \citet{Leoni2022} through a \emph{bandwagon mechanism}: individuals’ inclination to enrol increases with the observed prevalence of enrolment among peers, while other groups may display countervailing tendencies towards distinction or non-conformity. Related theoretical approaches to individual schooling choices include dynamic, generational (OLG-style) models where borrowing constraints and indivisibilities can trap families in low-education equilibria and generate multiple steady states \citep{Galor1993, DeLaCroix2003}, and a public-economics strand that studies optimal policy in \emph{static} equilibrium setups: partial-equilibrium household-production models that justify quantity controls like compulsory schooling \citep{Balestrino2017}, sectoral general-equilibrium school-choice models with endogenous tuition and peer composition \citep{Epple2002}, and economy-wide general-equilibrium models with endogenous skill prices under SBTC that characterize the optimal mix of linear income taxes and education subsidies \citep{Jacobs2022}. 

By contrast, in this paper we focus on a bandwagon–snob dynamic à la \citet{Caravaggio2020}, offering an explicitly dynamic account of how the wage premium and peer pressure jointly shape aggregate enrolment. This approach provides a compact representation of social diffusion processes and allows us to study analytically how the interplay between economic incentives and social imitation may produce multiple equilibria, cyclical fluctuations, or abrupt shifts in aggregate participation—phenomena reminiscent of the endogenous instability discussed by \citet{Benhabib1981}, who showed that even rational choices in a stationary environment can generate erratic, non-convergent trajectories when preferences evolve through experience.

Our contribution is twofold: first, we provide a formal theoretical analysis of how the interaction between expected economic returns to education (the \emph{wage premium}) and social influence (\emph{peer pressure}) can generate endogenous oscillations in aggregate educational cycles that emerge from internal feedback rather than exogenous shocks; second, we extend the scope of bandwagon--snob models to the domain of educational decisions, highlighting that the same socially interactive forces governing collective consumption behaviour may also underlie collective educational patterns, thereby building a conceptual bridge between the economics of education and the non-linear dynamics of social interaction. 

The remainder of the paper is organised as follows: \Cref{sec:model} introduces the basic dynamic model, describing how individual educational choices depend jointly on economic motivations and social interactions; \Cref{sec:result} presents the analytical and numerical exploration of the model with fixed population shares, focusing on the conditions under which imitation and distinction give rise to regular or irregular collective dynamics while \Cref{sec:endogenous} considers the more general case with endogenous shares. \Cref{sec:conclusion} concludes by outlining possible extensions and policy implications.

 \section{The Model}
 \label{sec:model}

We consider an economy where a unit mass of agents is split into two behavioural types: followers ($F$) and positional agents ($P$) where $\lambda\in(0,1)$ is the fixed share of conformists. Each agent has to decide the degree of education that they want to achieve ($e$) and, alternatively, can consume a generic consumption good $c$.

In each period each agent has an income $I$, while $p_e$ is the unit cost of education (that includes tuition and accessory costs) and $p_c$ is the cost of the consumption good. We assume that income and prices are constant over time. The budget constraint is therefore equal to:  
\begin{equation}
p_e\,e_{i,t}+p_c \,c_{i,t}= I .
\end{equation}

We model agents' preferences with a Cobb--Douglas utility function with time and type endogenous weights.

\begin{equation}
U_{i,t}(e_{i,t},c_{i,t};\alpha_{i,t},\beta_{i,t})=\alpha_{i,t}\log e_{i,t}+\beta_{i,t}\log c_{i,t},
\quad i\in\{F,P\}.
\end{equation}
Solving the constrained utility maximisation problem we get:

\begin{equation}
\label{eq:demand}
e^{*}_{i,t}
=\frac{\alpha_{i,t+1}}{\alpha_{i,t+1}+\beta_{i,t+1}} \cdot \frac{I}{p_e};
\qquad
c^{*}_{i,t} =\frac{\beta_{i,t+1}}{\alpha_{i,t+1}+\beta_{i,t+1}} \cdot \frac{I}{p_c},
\qquad i\in\{F,P\}.
\end{equation}

Starting from the individual demand we can derive the aggregate level of education and consumptions as:

\begin{align}
\label{eq:aggregate}
E_t=\lambda\,e_{F,t}+(1-\lambda)\,e_{P,t},\qq{where} E_t\in\Big[0,\frac{I}{p_e}\Big].\\
C_t=\lambda\,c_{F,t}+(1-\lambda)\,c_{P,t}, \qq{where} C_t\in\Big[0,\frac{I}{p_c}\Big].
\end{align}

As common in this literature, we assume that agents' preferential parameters depends on the demand of education and consumption from the previous period. However, education is not a standard consumption good. It possesses a dual nature: it is both a consumption good that provides immediate utility and an investment good in human capital, which can increase future earnings. We model this economic motivation by introducing a wage premium for obtaining a higher level of education. We assume this premium is endogenous and negatively correlated with the aggregate level of education ($E_t$). This captures a standard labour supply dynamic: the premium is high when educated workers are scarce and decreases as the supply of educated labour increases. We model the wage premium $\Pi_t$ in a simple way:

\begin{equation}
\Pi_t=\max(\, \overline{\Pi} - \kappa \, E_t; \; 0 \,) \qquad \overline{\Pi}>0, \, \kappa\ge 0,
\end{equation}
where $\overline{\Pi}$ is the highest possible wage premium, while $\kappa$ is a sensitivity parameter measuring the downward adjustment of the wage premium at the increase in supply of educated workers.

We assume that while followers are motivated by the imitative behaviour, the positional agents also care about the wage premium. Following \citep{Caravaggio2020}, the preferential parameters are defined as follows: 
\begin{equation}
\label{eq:parameters}
\begin{alignedat}{2}
\alpha_{F,t+1} &= 1 - \exp \left( - \rho \,E_t - \rho_\pi \Pi_t \right) \qquad
&\beta_{F,t+1} &= 1 - \exp\bigl(-\rho\, C_t\bigr) \\
\alpha_{P,t+1} &= \exp\bigl(-\sigma E_t\bigr)\,\bigl(1-\exp(-\sigma_{\Pi}\Pi_t)\bigr) \qquad
&\beta_{P,t+1} &= \exp\bigl(-\sigma\, C_t\bigr)
\end{alignedat}
\end{equation}

Since the aim of the analysis is to study educational choices, we can reduce the problem's dimensionality. The aggregate budget constraint creates a fixed linear trade-off between the two goods:
\begin{equation}
C_t = \frac{I}{p_c} - \frac{p_e}{p_c} E_t,
\end{equation}
Thus, the dynamics of aggregate consumption $C_t$ are a mere reflection of the dynamics of aggregate education $E_t$. It is therefore possible, and sufficient, to limit the entire analysis to the one-dimensional map for $E_t$.

Thus, by combining \cref{eq:aggregate,eq:demand} we obtain:





\begin{equation}
E_{t+1}
=\frac{I}{p_e}\left[
\lambda\,\frac{\alpha_{F,t+1}}{\alpha_{F,t+1}+\beta_{F,t+1}}+(1-\lambda)\,\frac{\alpha_{P,t+1}}{\alpha_{P,t+1}+\beta_{P,t+1}}
\right].
\end{equation}

After substituting \cref{eq:parameters} we obtain the following Map $\Gamma$:
\begin{equation}
\begin{aligned}
\Gamma:E_{t+1} =\frac{I}{p_c} \cdot \Bigg[\lambda \frac{1 - \exp(-\rho E_t - \rho_\pi \Pi_t)}{2-\exp(-\rho\,E_t - \rho_\pi \Pi_t)-\exp(-\frac{\rho(I-p_e\, E_t)}{p_c})} \hspace{6em}\\
 +(1-\lambda)\frac{\exp({-\sigma\, E_t}) \cdot (1- \exp(-\sigma_\pi \Pi_t))}{\exp({-\sigma \, E_t}) \cdot (1- \exp(-\sigma_\pi \Pi_t) + \exp({-\frac{\sigma(I-p_e E_t)}{p_c}})} \Bigg].
\label{eq:mappa}
\end{aligned}
\end{equation}

\section{Dynamic Trajectories with fixed population structure}
\label{sec:result}

We start our analysis by proving that the map admits at least one fixed point. 
\begin{proposition}[Existence of the fixed point]
	Let $\Gamma$ be the function defined in (\ref{eq:mappa}) and assume for simplicity that $\rho=\rho_\pi$ and $\sigma=\sigma_\pi$. Then:
    \begin{enumerate}[label=\textbf{(\roman*)}]
        \item if $p_e \leq p_c$ at last one fixed point always exists
        \item if $p_e > p_c$ a sufficient condition for the existence of at least one fixed point is 
        $$\lambda < \frac{p_c/p_e - S}{1 - S}  \qq{with} S = \frac{\exp(\sigma \frac{I}{p_e}(1 - \kappa) - \sigma \pi)}{1 + \exp(\sigma \frac{I}{p_e}(1 - \kappa) - \sigma \pi)}.$$
    \end{enumerate}
	\label{atleast}
\end{proposition}

\begin{proof}
	By calculating the value assumed by $\sigma$ at the extremes of the interval $[0,\frac{I}{p_e}]$, i.e. $\Gamma(0)$ and $\Gamma(\frac{I}{p_e})$, we have that:
	
	\begin{equation*}
	\Gamma(0)=\frac{I}{p_c}\left[\frac{\lambda(1-\exp(-\rho_\pi \bar{\Pi}))}{2 -\exp(-\rho_\pi \bar{\Pi})- \exp(-\rho I/p_c)}
    +\frac{ (1-\lambda)(1-\exp(-\sigma_\pi \bar{\Pi}))}{1-\exp(-\sigma_\pi \bar{\Pi}) + \exp(-\sigma I/p_c)}\right] \geq 0.
	\end{equation*}
    By solving $\Gamma(I/p_e) \leq 0$ we obtain 
	\begin{equation*}
	\lambda + (1-\lambda) S \leq \frac{I}{p_e}
	\end{equation*}
	This last condition is always satisfied if either condition $(i)$ or $(ii)$ are satisfied. 
\end{proof}

This result establishes the conditions under which at least one fixed point necessarily exists. However, due to the nonlinearity of the map, it is not possible to perform a detailed analytical study of the number and stability of the system’s fixed points. 

Nevertheless, the trajectories of the system can be characterised by means of the following result:

\begin{proposition}[Absorbing interval for $\Gamma$]\label{prop:absorbing} 
Let $D=[0,\bar E]\subset\mathbb{R}$ with $\bar E=I/p_e>0$ and assume the following conditions for the map $\Gamma:D\to D$, whose specification is given by \cref{eq:mappa}:
\begin{enumerate}[label=(H\arabic*)]
    \item $\Gamma$ is continuous on $D$ and $\Gamma(D)\subseteq D$.
    \item There exists a unique $E_c\in(0,\bar E)$ such that 
    \[
    \Gamma'(E)>0 \text{ for } E\in[0,E_c),\qquad
    \Gamma'(E)<0 \text{ for } E\in(E_c,\bar E],
    \]
    and $\Gamma(E)\le\Gamma(E_c)$ for all $E\in D$.
\end{enumerate}
By defining $
E_{\max}:=\Gamma(E_c)$, $E_{\min}:=\Gamma^2(E_c)=\Gamma(\Gamma(E_c))$ and $ J_\Gamma:=[E_{\min},E_{\max}]\subseteq D$.
Then the following properties hold:
\begin{enumerate}[label=(\roman*)]
    \item $\Gamma(J_\Gamma)\subseteq J_\Gamma$;
    \item $\Gamma^2(D)\subseteq J_\Gamma$. 
    Consequently, for every $E_0\in D$ there exists $n\le 2$ such that 
    $\Gamma^n(E_0)\in J_\Gamma$, and therefore $\Gamma^k(E_0)\in J_\Gamma$ for all $k\ge n$.
\end{enumerate}
\end{proposition}

\begin{proof} \phantom{dns}\\
Since $\Gamma$ is strictly increasing on $[0,E_c]$ and strictly decreasing on $[E_c,\bar E]$, we have
\[
\Gamma([0,E_c])=[\Gamma(0),E_{\max}],\qquad 
\Gamma([E_c,\bar E])=[\Gamma(\bar E),E_{\max}].
\]
Hence
\begin{equation}\label{eq:image_D}
\Gamma(D)=\Gamma([0,\bar E])=[m,E_{\max}], \qquad 
m:=\min\{\Gamma(0),\Gamma(\bar E)\}.
\end{equation}

\medskip\noindent
\textbf{Step 1: Forward invariance of $J_\Gamma$.}\\
By definition, $J_\Gamma=[\Gamma(E_{\max}),E_{\max}]=[\Gamma^2(E_c),\Gamma(E_c)]$.
Because $\Gamma$ is strictly decreasing on the right branch $[E_c,\bar E]$ and 
$E_{\max}\in\Gamma([E_c,\bar E])$ by~\eqref{eq:image_D}, 
there exists a unique $x_R\in[E_c,\bar E]$ such that
\[
\Gamma([E_c,x_R])=J_\Gamma,\qquad 
\Gamma(x_R)=\min J_\Gamma=E_{\min},\quad 
\Gamma(E_c)=\max J_\Gamma=E_{\max}.
\]
Therefore,
\[
\Gamma(J_\Gamma)
=\Gamma\big(\Gamma([E_c,x_R])\big)
=(\Gamma\circ\Gamma)([E_c,x_R]).
\]
Since $\Gamma$ is continuous on $[E_c,x_R]$, the image $(\Gamma\circ\Gamma)([E_c,x_R])$ 
is the interval between the extreme values of $\Gamma\circ\Gamma$ on this set.
Using the monotonicity of $\Gamma$:
\[
\begin{array}{l}
\min_{y\in[E_c,x_R]}(\Gamma\circ\Gamma)(y)
=\Gamma(\Gamma(E_c))=\Gamma(E_{\max})=E_{\min} \\
\max_{y\in[E_c,x_R]}(\Gamma\circ\Gamma)(y)
=\Gamma(\Gamma(x_R))=\Gamma(E_{\min})\le E_{\max},
\end{array}
\]
where the last inequality holds since $E_{\max}$ is the global maximum of $\Gamma$ on $D$.
Hence
\[
\Gamma(J_\Gamma)\subseteq [E_{\min},E_{\max}]=J_\Gamma,
\]
which proves forward invariance.

\medskip\noindent
\textbf{Step 2: Absorption in at most two iterates.} \\
From~\cref{eq:image_D}, $\Gamma(D)=[m,E_{\max}]\subseteq [0,E_{\max}]$.
Applying $\Gamma$ once more and using Step~1 on any interval whose right endpoint is $E_{\max}$, we obtain:
\[
\Gamma([0,E_{\max}])\subseteq [\Gamma(E_{\max}),E_{\max}]
=[E_{\min},E_{\max}]=J_\Gamma.
\]
Therefore $\Gamma^2(D)\subseteq J_\Gamma$. 
Therefore, for each $E_0\in D$ there exists $n\le 2$ such that $\Gamma^n(E_0)\in J_\Gamma$, 
and by forward invariance all subsequent iterates remain in $J_\Gamma$.
\end{proof}

The result indicates that when the dynamical system is self-mapped and unimodal, all trajectories remain confined within a bounded region of the phase plane. Whether they converge to an equilibrium point, a stable periodic orbit, or a chaotic attractor, they are all contained within a specific subset of the plane known as the \emph{absorbing region} (or \emph{trapping region}).


\subsection{No preference for wage premium}

In this section we investigate the case where only the positional agents have a preference for the wage premium. This corresponds to the case where $\rho_\pi =0$. 

Given the high non-linearity of $\Gamma(E_t)$, resulting from the exponential preference functions and the interaction between agent types, an analytical study of the map's fixed points and their stability is intractable. We therefore rely on numerical simulations to characterize the system's long-run behaviour.

Our analysis focuses on two primary factors: (A) the behavioural reactivity of the two groups, and (B) the population composition, $\lambda$. Factor (A) concerns the interplay between the imitative drive of Followers (governed by $\rho$) and the counter-adaptive positional drive of Positional Agents (governed by $\sigma_E$, $\sigma_{\pi}$, and the premium's sensitivity $\kappa$). Factor (B) explores how the population mix itself influences aggregate stability.

Starting from the behavioural reactivity, \Cref{fig:bifurcation_paramenters} shows these dynamics. Panel (a) shows the bifurcation diagram with respect to $\rho$ (holding all other parameters constant) while Panel (b) shows the bifurcation diagram with respect to $\sigma$.
\FloatBarrier
\begin{figure}[htbp]
  \centering
  \begin{subfigure}[b]{0.48\textwidth}
    \centering
    \includegraphics[width=\textwidth]{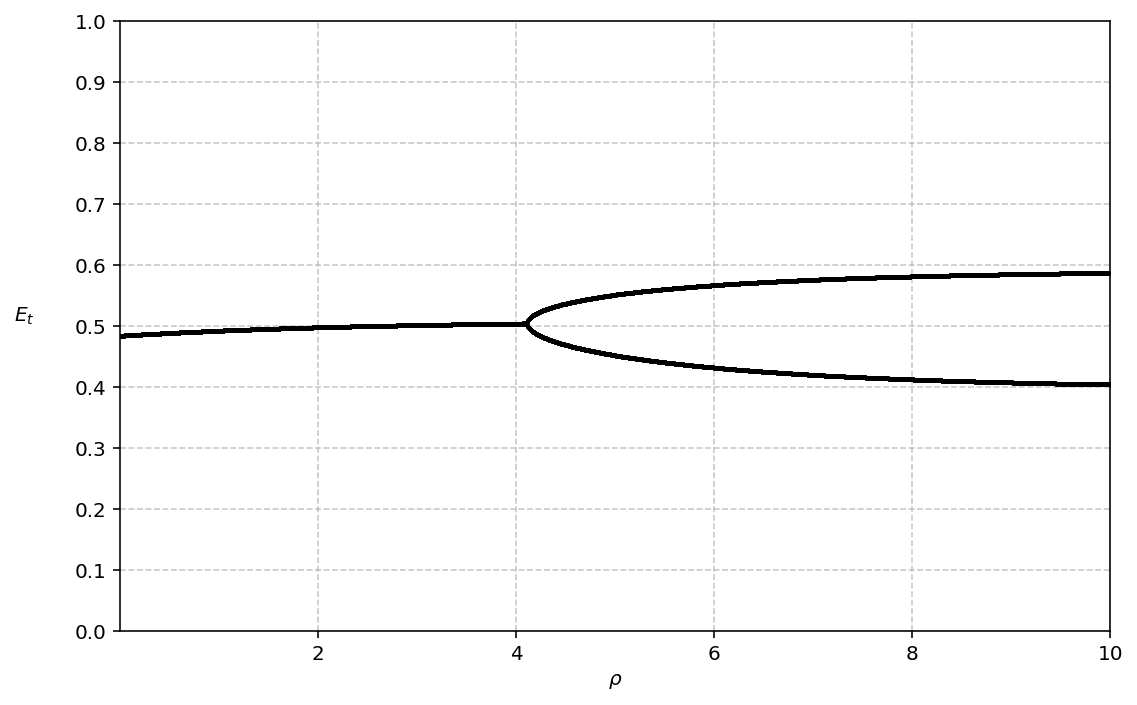}
    \caption{Bifurcation Diagram with respect to $\rho$}
    \label{fig:A}
  \end{subfigure}
  \hfill
  \begin{subfigure}[b]{0.48\textwidth}
    \centering
    \includegraphics[width=\textwidth]{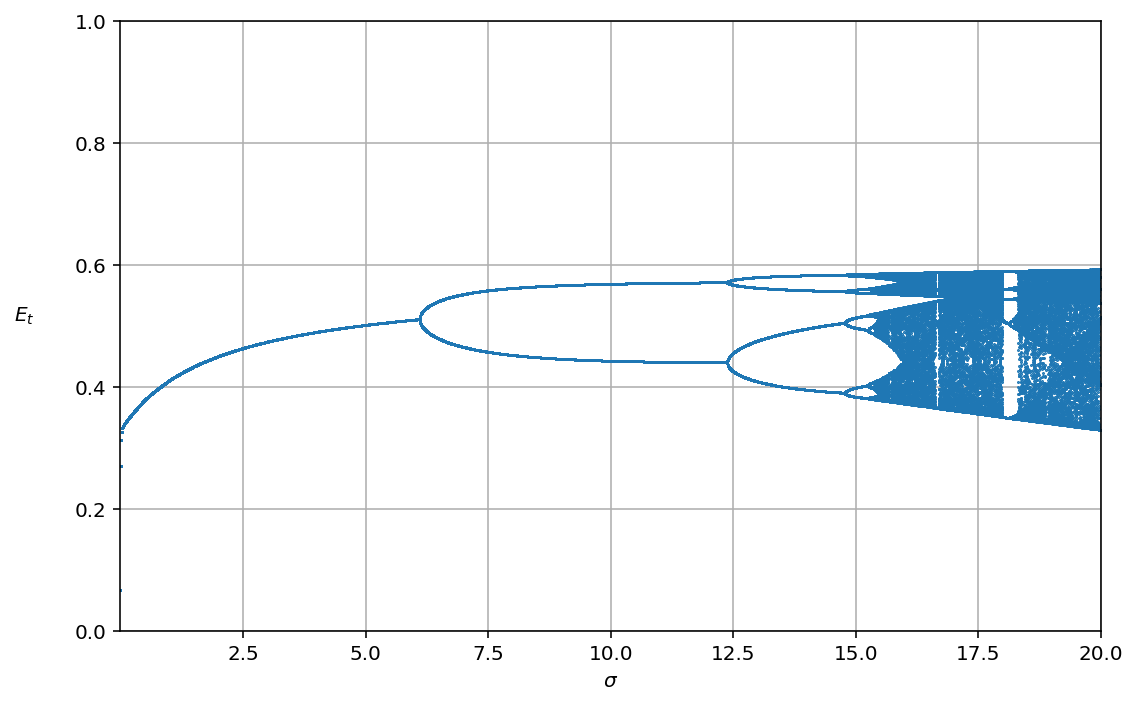}
    \caption{Bifurcation Diagram with respect to $\sigma$}
    \label{fig:B}
  \end{subfigure}
  \caption{Bifurcation Diagrams of behavioural reactivity. Parameter set: $\rho = 0.98$, $I = 1.0$, $p_e = 1.2$, $p_c = 0.53$, $\lambda = 0.5$, $\kappa = 0.3$, $\overline{\Pi}= 100$. }
  \label{fig:bifurcation_paramenters}
\end{figure}
\FloatBarrier


\Cref{fig:bifurcation_paramenters} reveals a crucial asymmetry in the destabilizing power of the two agent groups. Holding the Followers' imitative reactivity ($\rho$) constant, we observe that the Positional Agents' counter-adaptive sensitivity (e.g., $\sigma_E$) acts as a powerful destabilizing force, as shown in Panel (b). Starting from a stable steady state, the system undergoes its first flip bifurcation at approximately $\sigma \approx 5.8$. As this parameter increases further, it triggers a full period-doubling cascade (or "Feigenbaum route"), rapidly pushing the system into a chaotic regime for $\sigma \gtrsim 16$. This confirms that the "snobbish" drive for distinction, combined with the negative feedback from the wage premium, is a primary engine of complex dynamics.Conversely, the role of the Followers' imitative drive ($\rho$) is substantially different. Holding the positional parameters constant, increasing $\rho$ does amplify the system's non-linearity, but its destabilizing effect is far more contained. As shown in Panel (a),  even for very high values of $\rho$, the system at most undergoes a single flip bifurcation (at $\rho \approx 4.1$), resulting in a stable 2-cycle that persists without further bifurcations.This key finding suggests that while pro-cyclical imitation ($\rho$) can induce simple, regular oscillations, it is the strong counter-adaptive force ($\sigma_E$, $\sigma_{\Pi}$) of the Positional Agents that is ultimately responsible for generating complex, erratic, and chaotic trajectories in the aggregate demand for education.

To better visualize the nature of these distinct dynamic regimes, we now turn to a direct analysis of the one-dimensional map, $E_{t+1} = H(E_t)$, and its corresponding time series. 

\FloatBarrier
\begin{figure}[htbp]
  \centering
  \includegraphics[scale=0.5]{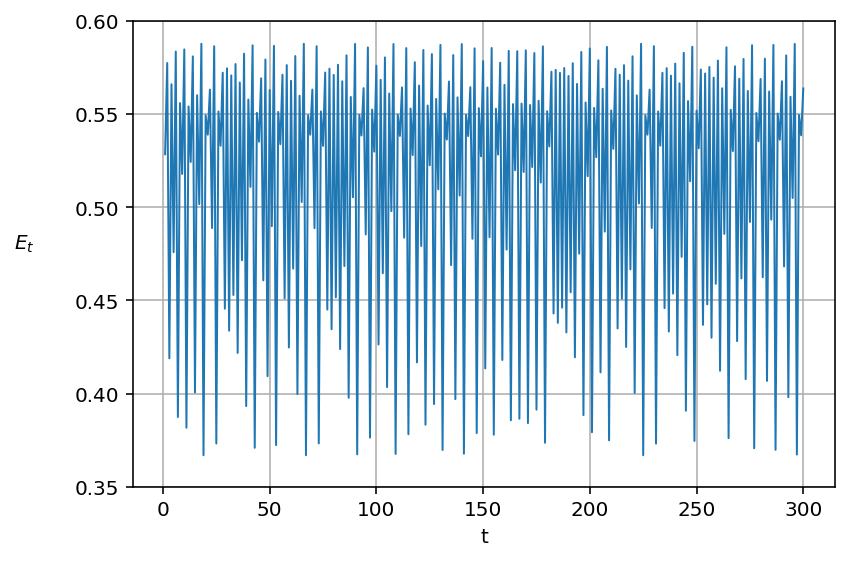}
  \caption{Time series of $E$. {Parameter set: $\rho = 0.98$, $I = 1.0$, $p_e = 1.2$, $p_c = 0.53$, $\lambda = 0.5$, $\kappa = 0.3$, $\overline{\Pi}= 100$, $\sigma = 16.5$.}}
  \label{fig:time_series}
\end{figure}
\FloatBarrier

\Cref{fig:time_series} displays the system's dynamics within the chaotic regime, generated using a high positional reactivity value of $\sigma = 16.5$. This parameter setting corresponds to the chaotic attractor identified in the bifurcation diagram (\Cref{fig:bifurcation_paramenters}).
The panel plots the aggregate education level, $E_t$, over 300 time periods. The trajectory clearly demonstrates the defining characteristics of deterministic chaos: the system fails to converge to either a stable fixed point or a regular periodic cycle.

Instead, $E_t$ exhibits persistent, aperiodic, and bounded oscillations. The fluctuations are highly irregular and erratic, yet they remain confined within a specific range (in this simulation, approximately $E_t \in [0.38, 0.59]$). This reflects the undesirable policy outcome of endogenous instability. In such a regime, the education market is subject to perpetual, unpredictable volatility, making rational long-term planning for both students and institutions practically impossible.

\FloatBarrier
\begin{figure}[htbp]
  \centering
  \includegraphics[scale=0.6]{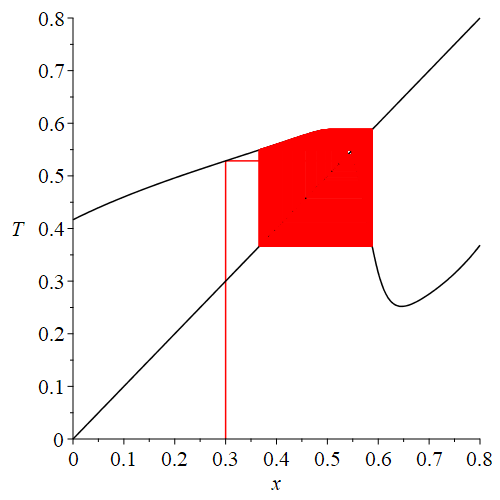}
  \caption{Cobweb Diagram of $E_t$}
  \label{fig:ragnatela}
\end{figure}
\FloatBarrier

\Cref{fig:ragnatela} provides the geometric explanation for the chaotic dynamics observed in the time series. This cobweb diagram plots the one-dimensional map $E_{t+1} = \Gamma(E_t)$ (the black curve, with $E_{t+1}$ labeled as $T$) against the 45-degree bisector (the $E_{t+1} = E_t$ identity line).

The figure clearly illustrates the multimodality (i.e., the complex, non-monotonic, "humped" shape) of the map $H(E_t)$. This non-linearity is the fundamental driver of the complex dynamics, as it is a direct result of the conflicting pro-cyclical (Follower) and counter-cyclical (Positional) forces present in the model.

The large, red-filled area highlights the chaotic attractor, or trapping region. This is the bounded subset of the phase space to which the system is confined in the long run. The red lines trace the path of a sample trajectory starting from an initial condition of $E_0 \approx 0.3$. The system is immediately pulled into this attractor after the first iteration.

Once inside this region, the map's non-monotonicity causes the trajectory to "stretch and fold" back onto itself, ensuring that the dynamics are both non-convergent and aperiodic. This confirms that the aggregate education level will never settle into a stable equilibrium or a simple, predictable cycle.

Finally, we analyse the role of the population structure, $\lambda$. The bifurcation diagram in \Cref{fig:bifurcation_lambda} shows the system's long-run dynamics as $\lambda$ varies from $0$ (only Positional Agents) to $1$ (only Followers), holding the reactivity parameters fixed at values that *could* generate chaos.

\begin{figure}[htbp]
  \centering
  \includegraphics[scale=0.4]{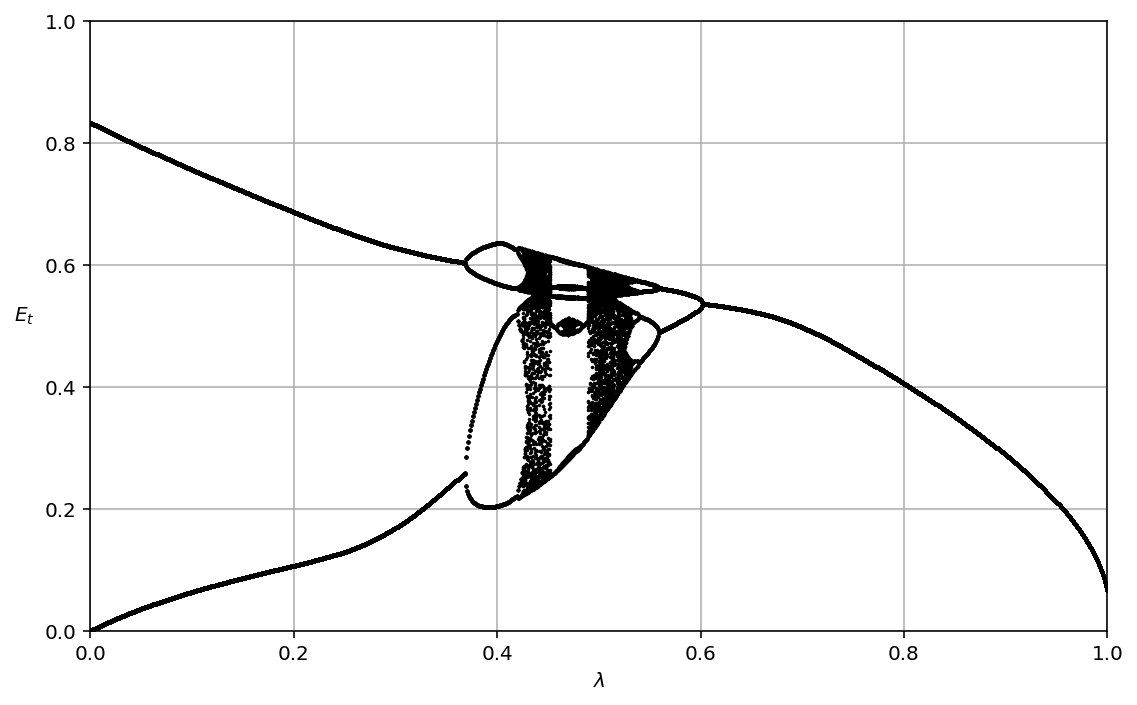}
  \caption{Bifurcation Diagram with respect to $\lambda$. {Parameter set: $\rho = 1.18$, $I = 1.0$, $p_e = 1.2$, $p_c = 0.53$, $\kappa = 0.3$, $\overline{\Pi}= 100$, $\sigma = 17.9$.}}
  \label{fig:bifurcation_lambda}
\end{figure}
\FloatBarrier

The effect of $\lambda$ on stability is complex and non-monotonic. When the population is dominated by Positional Agents ($\lambda \to 0$), the dynamics are stable or exhibit simple cycles. Similarly, when the population is dominated by Followers ($\lambda \to 1$), the strong imitative consensus leads to a stable steady state. The most complex dynamics—including the chaotic regime—emerge only for intermediate, heterogeneous populations (e.g., $\lambda \in [0.2, 0.6]$ in our example).

This result highlights a key finding: the emergence of complex cycles in education demand is not a feature of either group in isolation. Rather, it is an emergent property of the social conflict between the imitative, conforming behaviour of Followers and the counter-adaptive, status-seeking behaviour of Positional Agents.

This implies that a certain degree of social heterogeneity is the key driver of the most unstable and unpredictable outcomes. From a policy perspective, this presents a significant challenge: it suggests that the most difficult planning problems (e.g., funding, labour market forecasting) arise precisely in diverse societies where these conflicting social behaviours coexist, rather than in more behaviourally homogeneous ones.
\FloatBarrier

\subsection{Preference for wage premium}
We now introduce a crucial variation where Followers also develop a preference for the wage premium (i.e., $\rho_{\Pi} > 0$). This fundamentally alters their behaviour: their purely imitative (pro-cyclical) drive is now dampened by a counter-cyclical economic concern. When enrolment ($E_t$) is high, the resulting low wage premium ($\Pi_t$) now provides a negative feedback for both groups, reducing the overall behavioural conflict.

This shared economic rationality acts as a stabilizing force. As shown in \Cref{fig:ciclo2}, we use the identical parameter set ($\sigma = 16.5$) that previously generated chaos. With $\rho_{\Pi} > 0$, the chaotic attractor collapses, and the system's dynamics simplify dramatically into a stable, regular 2-cycle.

This result  suggests the most severe instabilities (chaos) are not driven by positional concerns alone, but by the sharp conflict between pure imitation and positional, counter-adaptive behaviour. A shared economic signal aligns incentives and restores a predictable, cyclical pattern.


\FloatBarrier
\begin{figure}[htbp]
  \centering
  \begin{subfigure}[b]{0.48\textwidth}
    \centering
    \includegraphics[width=\textwidth]{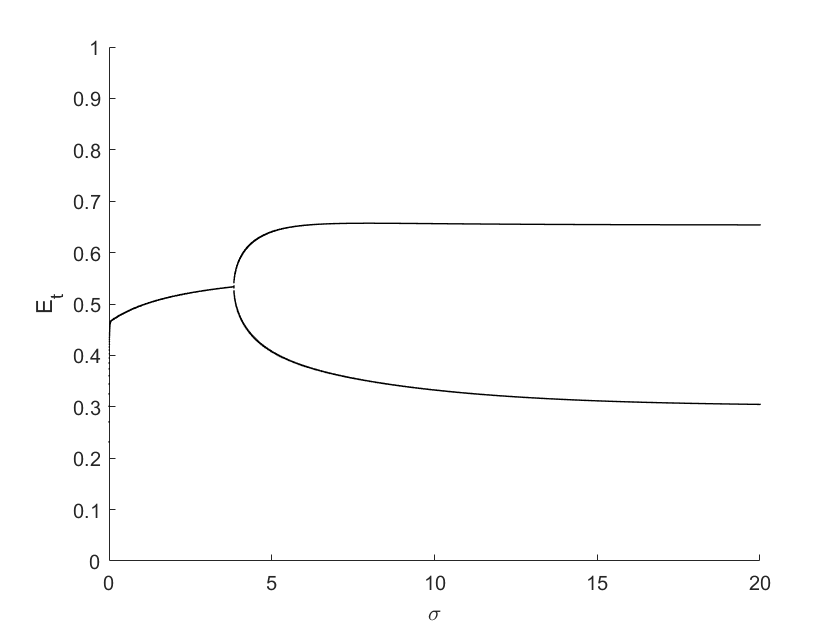}
    \caption{Bifurcation Diagram with respect to $\rho$}
    \label{subfig:bifurcation_2}
  \end{subfigure}
  \hfill
  \begin{subfigure}[b]{0.48\textwidth}
    \centering
    \includegraphics[scale = 0.55]{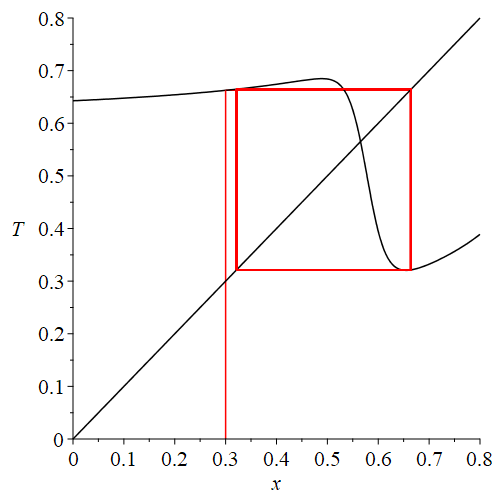}
    \caption{Bifurcation Diagram with respect to $\sigma$}
    \label{subfig:cobweb_2}
  \end{subfigure}
  \caption{Bifurcation Diagrams of behavioural reactivity. Parameter set: $\rho = 0.98$, $I = 1.0$, $p_e = 1.2$, $p_c = 0.53$, $\lambda = 0.5$, $\kappa = 0.3$, $\overline{\Pi}= 100$. }
  \label{fig:ciclo2}
\end{figure}
\FloatBarrier

\subsection{Effects of sensitivity to increases of supply of educated workers}
Finally the next proposition characterises the study of comparative statics on the effect of a change of the sensitivity of the increase of the supply of educated workers on the wage premium. This result is derived by means of the implicit function theorem.
\begin{proposition}
\label{prop:comparativestat_kappa}
Let $\Gamma(E;\kappa)$ denote the map in equation~(10), where the wage premium is 
$\Pi_t=\Pi-\kappa E_t>0$. 
Let $E^\star\in(0,\bar E)$ be an interior fixed point satisfying
\begin{equation}
F(E,\kappa)\equiv \Gamma(E;\kappa)-E=0, \qquad \bar E=\frac{I}{p_e},
\end{equation}
and assume local stability, $|\Gamma_E(E^\star;\kappa)|<1$.
Then:

\begin{enumerate}[label=(\roman*)]
\item There exists a continuously differentiable function 
$E^\star(\kappa)$ such that $F(E^\star(\kappa),\kappa)=0$, and
\begin{equation}
\frac{dE^\star}{d\kappa}
=\frac{\Gamma_\kappa(E^\star;\kappa)}{\,1-\Gamma_E(E^\star;\kappa)\,}.
\label{eq:dEdkappa}
\end{equation}

\item In the regime $\Pi^\star=\Pi-\kappa E^\star>0$, the map $\Gamma$ is increasing in $\Pi$ (i.e.\ $\Gamma_\Pi>0$) and 
$\partial \Pi/\partial \kappa=-E<0$. Hence
\[
\Gamma_\kappa=\Gamma_\Pi\frac{\partial\Pi}{\partial\kappa}
=-E\,\Gamma_\Pi<0,
\]
and since local stability implies $1-\Gamma_E>0$, it follows that
\begin{equation}
\frac{dE^\star}{d\kappa}<0.
\end{equation}
An increase in the wage-premium sensitivity $\kappa$ reduces the long-run enrolment level $E^\star$.
\end{enumerate}

\noindent

\end{proposition}

\begin{proof}[Proof]
Within the interior regime $\Pi_t=\Pi-\kappa E_t$, the map $\Gamma$ is smooth in both arguments. 
Applying the Implicit Function Theorem to $F(E,\kappa)=0$ directly yields~\eqref{eq:dEdkappa}. 
The sign of $\tfrac{dE^\star}{d\kappa}$ depends on $\Gamma_\kappa$, which can be evaluated via $\Gamma_\Pi$.

\smallskip
\noindent
Let $A=-\rho E-\rho_\pi \Pi$ and $K=\exp\!\left(-\rho\frac{I-p_eE}{p_c}\right)\in(0,1]$. 
The fraction
\[
f(\Pi)=\frac{1-e^{A}}{2-e^{A}-K}
\]
appearing in the followers’ component of~(10) satisfies $f_\Pi>0$. 
Indeed, setting $x=e^{A}$ gives $f(x)=\frac{1-x}{2-x-K}$ with 
$f'(x)=\frac{K-1}{(2-x-K)^2}\le0$ and $dx/d\Pi=-\rho_\pi e^{A}<0$, 
so that $df/d\Pi\ge0$.

\noindent For the positional part of~(10), write
\[
g(\Pi)=\frac{a\,b}{a\,b+k},\qquad 
a=e^{-\sigma E},\;
b=1-e^{-\sigma_\pi\Pi},\;
k=e^{-\sigma\frac{I-p_eE}{p_c}}.
\]
Then $\partial g/\partial b=\tfrac{a\,k}{(a\,b+k)^2}>0$ and 
$db/d\Pi=\sigma_\pi e^{-\sigma_\pi\Pi}>0$, so $g_\Pi>0$.
\noindent Both components of $\Gamma$ are increasing in $\Pi$, implying $\Gamma_\Pi>0$. 
Since $\partial\Pi/\partial\kappa=-E$, one has $\Gamma_\kappa<0$. 
For a stable fixed point, $1-\Gamma_E>0$, and hence 
$\tfrac{dE^\star}{d\kappa}<0$.
\end{proof}
This result is consistent with the economic intuition.\footnote{If the fixed point lies in the saturation regime $\Pi^\star=0$, then $\Gamma$ does not depend on $\kappa$ locally and $\tfrac{dE^\star}{d\kappa}=0$ as long as the equilibrium remains in that region.
At the transition $\Pi^\star=0$ (regime switch) or when $|\Gamma_E|=1$ (bifurcation points), the Implicit Function Theorem does not apply and the dependence of $E^\star$ on $\kappa$ may become non-differentiable or discontinuous.} A higher sensitivity of the wage premium to the increase in skilled-labour supply, leads the a larger reduction of the wage premium for the same level of aggregate education. Thus, education will be less valuable, leading to a lower equilibrium value. 

\begin{remark}
The negative effect of $\kappa$ on $E^\star$ holds irrespective of whether followers’ preferences include the wage-premium term ($\rho_\pi>0$) or not, because all components depending on $\Pi$ are monotonically increasing in $\Pi$. 
\end{remark}

\section{Dynamic Trajectory with endogenous population structure}
\label{sec:endogenous}

The results shown in the previous section arise from a study where the population structure is fixed; the share of Followers ($\lambda$) and Positional Agents ($1-\lambda$) cannot change. It is reasonable to examine how the system's dynamics may change in a framework where individuals can move from one group to another through an endogenous mechanism. To this end, we extend the baseline model by endogenizing the population share 
of Followers, $\lambda_t\in[0,1]$. The aggregate state is 
now $(E_t,\lambda_t)\in D:=\big[0,\bar E\big]\times[0,1]$, with $\bar E=I/p_e$. Given the Cobb--Douglas utility $U_{i,t}(e,c;\alpha_{i,t},\beta_{i,t})=\alpha_{i,t}\log e+\beta_{i,t}\log c$,
the optimal demands at time $t$ (conditional on $\alpha_{i,t+1},\beta_{i,t+1}$) are:
\begin{equation}
e^{\ast}_{i,t}=\frac{\alpha_{i,t+1}}{\alpha_{i,t+1}+\beta_{i,t+1}}\cdot \frac{I}{p_e},
\qquad 
c^{\ast}_{i,t}=\frac{\beta_{i,t+1}}{\alpha_{i,t+1}+\beta_{i,t+1}}\cdot \frac{I}{p_c},
\qquad i\in\{F,P\}.
\label{eq:demands}
\end{equation}
Evaluated at \eqref{eq:demands}, the corresponding (type-specific) indirect utilities are:
\begin{equation}
\begin{aligned}
\mathcal{U}_{i,t}
&:=U_{i,t}\Big(e^{\ast}_{i,t},c^{\ast}_{i,t};\alpha_{i,t+1},\beta_{i,t+1}\Big) \\
&=\alpha_{i,t+1}\!\left[\log\!\left(\frac{I}{p_e}\right)+\log\!\left(\frac{\alpha_{i,t+1}}{\alpha_{i,t+1}+\beta_{i,t+1}}\right)\right]
+\beta_{i,t+1}\!\left[\log\!\left(\frac{I}{p_c}\right)+\log\!\left(\frac{\beta_{i,t+1}}{\alpha_{i,t+1}+\beta_{i,t+1}}\right)\right].
\end{aligned}
\label{eq:indirectU}
\end{equation}
Note that via \cref{eq:parameters}, both $\alpha_{i,t+1}$ and $\beta_{i,t+1}$ 
depend on $E_t$ (and hence on $\lambda_t$ only through $E_t$).\\ The aggregate dynamics are given by the map $\Phi:D\to D$ defined as:
\begin{equation}
\Phi:\quad 
\begin{cases}
E_{t+1} = \displaystyle \frac{I}{p_e}\left[
\lambda_t \frac{\alpha_{F,t+1}}{\alpha_{F,t+1}+\beta_{F,t+1}}
+ (1-\lambda_t)\frac{\alpha_{P,t+1}}{\alpha_{P,t+1}+\beta_{P,t+1}}
\right] 
=: \Gamma\!\big(E_t;\lambda_t\big),\\[1.2em]
\lambda_{t+1} = 
\displaystyle \frac{\exp\!\big(\mu\,\mathcal{U}_{F,t}\big)}
{\exp\!\big(\mu\,\mathcal{U}_{F,t}\big)+\exp\!\big(\mu\,\mathcal{U}_{P,t}\big)}
=: \mathcal{V}\!\big(E_t,\lambda_t\big),
\end{cases}
\label{eq:2Dsystem}
\end{equation}
where $\mu\ge 0$ measures the \emph{willingness to switch}. According to \citet{Naimzada2018}, it describes the ability of individuals to analyse and compare the benefits made to decide the type of behaviour to be adopted in the next period, i.e. their willingness to change and move towards the preferences that guaranteed the highest utility level in the previous period. In particular, we can notice that for $\mu\rightarrow\,0$ the individuals continue to choose the preferences they chose in the previous period, i.e. $\lambda_{t+1}=\lambda_{t}$; when $\mu>0$, the function in $\mathcal{V}\!\big(E_t,\lambda_t\big)$ is an increasing monotone function of $U_{F,t}(\alpha_{F,t},\beta_{F,t},e^{\ast}_{F,t},c^{\ast}_{F,t})-U_{P,t}(\alpha_{P,t},\beta_{P,t},e^{\ast}_{P,t},c^{\ast}_{P,t})$; as $\mu\rightarrow\,+\infty$, the individuals move in one shot toward the preferences that guaranteed the highest utility level in the previous period. Specifically, (i) if $U_{F,t}(\alpha_{F,t},\beta_{F,t},e^{\ast}_{F,t},c^{\ast}_{F,t})<U_{P,t}(\alpha_{P,t},\beta_{P,t},e^{\ast}_{P,t},c^{\ast}_{P,t})$, then $\omega^{'}\rightarrow\,0$ as $\mu\rightarrow\,+\infty$, (ii) if $U_{F,t}(\alpha_{F,t},\beta_{F,t},e^{\ast}_{F,t},c^{\ast}_{F,t})>U_{P,t}(\alpha_{P,t},\beta_{P,t},e^{\ast}_{P,t},c^{\ast}_{P,t})$, then $\omega^{'}\rightarrow\,1$ as $\mu\rightarrow\,+\infty$.\\
By construction, $\Gamma:[0,\bar E]\times[0,1]\to[0,\bar E]$ and 
$\mathcal{V}:[0,\bar E]\times[0,1]\to[0,1]$, so $D$ is forward invariant. On $D$, the map $\Phi$ is continuous (indeed, smooth) except possibly at the regime-switch frontier $\Pi_t=0$, where the premium saturates. For $\Pi_t>0$, $\Phi$ is $C^\infty$ in all arguments and parameters.

Let $F(E,\lambda):=\Gamma(E;\lambda)-E$ and $G(E,\lambda):=\mathcal{V}(E,\lambda)-\lambda$.
At an interior fixed point $(E^\star,\lambda^\star)$, the Jacobian is
\[
J(E^\star,\lambda^\star)=
\begin{pmatrix}
\Gamma_E & \Gamma_\lambda\\
\mathcal{V}_E & \mathcal{V}_\lambda-1
\end{pmatrix}_{(E^\star,\lambda^\star)}.
\]
Using \eqref{eq:parameters}--\eqref{eq:indirectU}, one gets
\[
\Gamma_\lambda = \frac{I}{p_e}\left[\frac{\alpha_{F}}{\alpha_{F}+\beta_{F}}-\frac{\alpha_{P}}{\alpha_{P}+\beta_{P}}\right],
\qquad
\mathcal{V}_\lambda = 0,\quad
\mathcal{V}_E = 
\mu\,\mathcal{V}(1-\mathcal{V})\cdot\big(\mathcal{U}_{F,E}-\mathcal{U}_{P,E}\big),
\]
where subscripts denote partial derivatives evaluated at $(E^\star,\lambda^\star)$ and 
\[
\mathcal{U}_{i,E}=
\alpha'_{i}\left[\log\!\left(\frac{I}{p_e}\right)+\log\!\left(\frac{\alpha_{i}}{\alpha_{i}+\beta_{i}}\right)+1-\frac{\alpha_{i}}{\alpha_{i}+\beta_{i}}\right]
+\beta'_{i}\left[\log\!\left(\frac{I}{p_c}\right)+\log\!\left(\frac{\beta_{i}}{\alpha_{i}+\beta_{i}}\right)+1-\frac{\beta_{i}}{\alpha_{i}+\beta_{i}}\right]
\]
where $\alpha'_i,\beta'_i$ collect the derivatives of \cref{eq:parameters} w.r.t.\ $E$.
Despite the complexity of the analysis we provide 2 proposition concerning the local stability of the fixed point of the map (assuming it exists). 
\begin{proposition}
\label{prop:local_stability}
Consider the two-dimensional map $\Phi(E,\lambda)=\big(\Gamma(E;\lambda),\,\mathcal V(E,\lambda)\big)$
with $\Gamma$ and $\mathcal V$. 
Let $(E^\star,\lambda^\star)$ be an \emph{interior} fixed point and the Jacobian matrix
\[
J(E^\star,\lambda^\star)=
\begin{pmatrix}
\Gamma_E & \Gamma_\lambda\\
\mathcal V_E & 0
\end{pmatrix}_{(E^\star,\lambda^\star)},
\qquad
\tau:=\mathrm{tr}\,J=\Gamma_E,\quad
\Delta:=\det J=-\,\Gamma_\lambda\,\mathcal V_E.
\]
Then $(E^\star,\lambda^\star)$ is (locally) asymptotically stable if and only if the Schur conditions hold:
\[
1-\Gamma_E-\Gamma_\lambda\mathcal V_E>0,\qquad
1+\Gamma_E-\Gamma_\lambda\mathcal V_E>0,\qquad
1+\Gamma_\lambda\mathcal V_E>0.
\]
\end{proposition}

\begin{proof}[Proof]
See the appendix.
\end{proof}

\begin{proposition}
\label{prop:mu-threshold}
Let $(E^\star,\lambda^\star)$ be an interior fixed point with $\Pi^\star>0$.
Define
\[
g^\star:=\big|\Gamma_E(E^\star,\lambda^\star)\big|,\qquad
h^\star:=\big|\mathcal U_{F,E}(E^\star)-\mathcal U_{P,E}(E^\star)\big|.
\]
Then a sufficient condition for local asymptotic stability is
\begin{equation}
\label{eq:l1-bound-main}
\quad |\Gamma_E|+|\Gamma_\lambda\,\mathcal V_E|<1\quad \text{at }(E^\star,\lambda^\star).
\end{equation}
Moreover, since $0\le \mathcal V(1-\mathcal V)\le \tfrac14$ and $|s_F-s_P|\le 1$, one has
\[
|\Gamma_\lambda\,\mathcal V_E|
\le \frac{I}{p_e}\cdot \frac{\mu}{4}\cdot h^\star,
\]
so that \eqref{eq:l1-bound-main} is guaranteed whenever
\begin{equation}
\label{eq:mu-threshold-main}
\quad \mu\;<\;\bar\mu\;:=\;\frac{4p_e}{I}\,\frac{1-g^\star}{\,h^\star\,}\,,\quad\text{with }\,0\le g^\star<1.
\end{equation}
Hence, any $\mu<\bar\mu$ ensures local asymptotic stability of $(E^\star,\lambda^\star)$.
\end{proposition}

\begin{proof}
Inequality \eqref{eq:l1-bound-main} implies the Schur conditions because, writing
$x:=\Gamma_E$ and $y:=-\Gamma_\lambda\mathcal V_E=\Delta$, the bound
$|x|+|y|<1$ yields $1-x+y>0$, $1+x+y>0$, and $1-y>0$.
Using $\mathcal V_E=\mu\,\mathcal V(1-\mathcal V)\big(\mathcal U_{F,E}-\mathcal U_{P,E}\big)$
and $\Gamma_\lambda=\tfrac{I}{p_e}(s_F-s_P)$, we obtain
$|\Gamma_\lambda\mathcal V_E|\le \tfrac{I}{p_e}\tfrac{\mu}{4}h^\star$ and thus the threshold \eqref{eq:mu-threshold-main}.
Further details are available in the Appendix.
\end{proof}
After having established the local stability conditions of the fixed point for the map $\Phi$ and shown that, for sufficiently low values of $\mu$, the fixed point can preserve its stability, we now focus on a numerical investigation of the possible emergence of chaotic dynamics in the two-dimensional system.

First, it can be observed that, also in the two-dimensional case, an increase in the reactivity of the snob agent may induce the onset of irregular trajectories, both with respect to the levels expressed by education and to the population structure. Specifically, for sufficiently high values of $\sigma$, oscillations can be observed that drive $\lambda$ alternately towards its extreme values (monomorphic population) and towards intermediate levels (see Panels (a) and (b) in the \Cref{fig:endo}).

Furthermore, it can be noted that the parameter regulating the willingness to change type also contributes to destabilizing the dynamics. Indeed, as illustrated in Panels (a) and (b) of \Cref{fig:mu}, even starting from a given (and sufficiently low) value of $\sigma$, increasing $\mu$ towards sufficiently high levels leads to the emergence of complex dynamics.

In conclusion, these numerical experiments indicate that, for extreme values of the parameters governing agents’ reactivity and the structure of the population, the dynamics of both education and population structure become trapped within a chaotic attractor.

\FloatBarrier
\begin{figure}[htbp]
  \centering
  \begin{subfigure}[b]{0.48\textwidth}
    \centering
    \includegraphics[width=\textwidth]{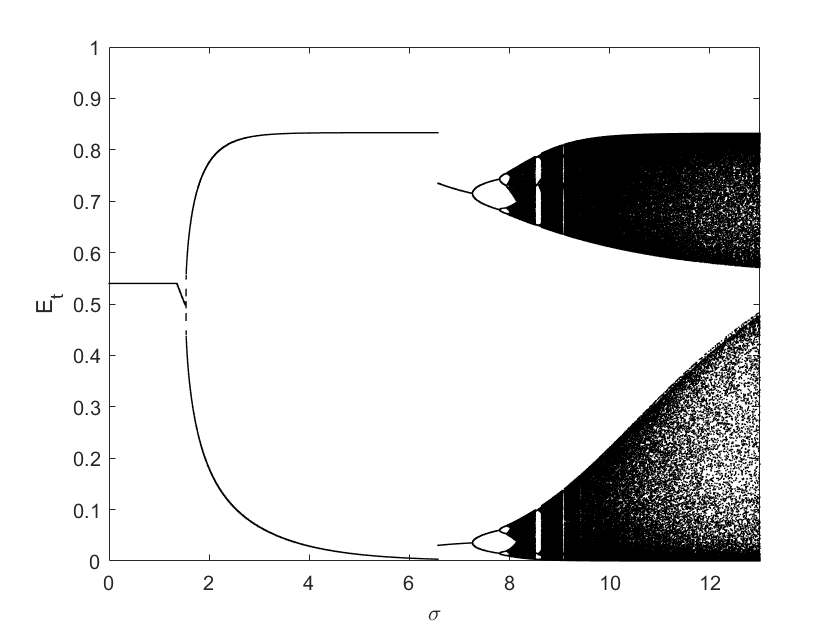}
    \caption{Bifurcation Diagram with respect to $\sigma$}
    \label{subfig:bifurcation_endo_sigma}
  \end{subfigure}
  \hfill
  \begin{subfigure}[b]{0.48\textwidth}
    \centering
    \includegraphics[width=\textwidth]{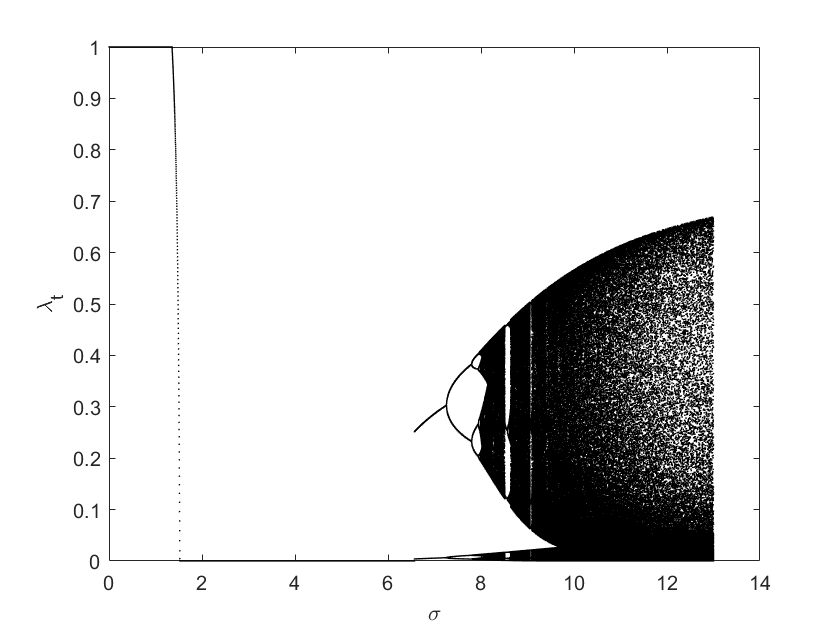}
    \caption{Bifurcation Diagram with respect to $\sigma$}
    \label{subfig:bifurcation_end_lambda}
  \end{subfigure}
  \caption{Bifurcation Diagrams of behavioural reactivity. Parameter set: $\rho = 0.98$, $I = 1.0$, $p_e = 1.2$, $p_c = 0.53$, $\kappa = 0.3$, $\overline{\Pi}= 100$. }
  \label{fig:endo}
\end{figure}
\FloatBarrier

In a two dimensional system the endogeneity of $\lambda$ leads to a propagation of chaos, both with respect to $\lambda$ and $\sigma$. 

Finally, it is worth showing that the actual propagator of the chaotic behaviour of this system is the role played by $\mu$. 

\FloatBarrier
\begin{figure}[htbp]
  \centering
  \begin{subfigure}[b]{0.48\textwidth}
    \centering
    \includegraphics[width=\textwidth]{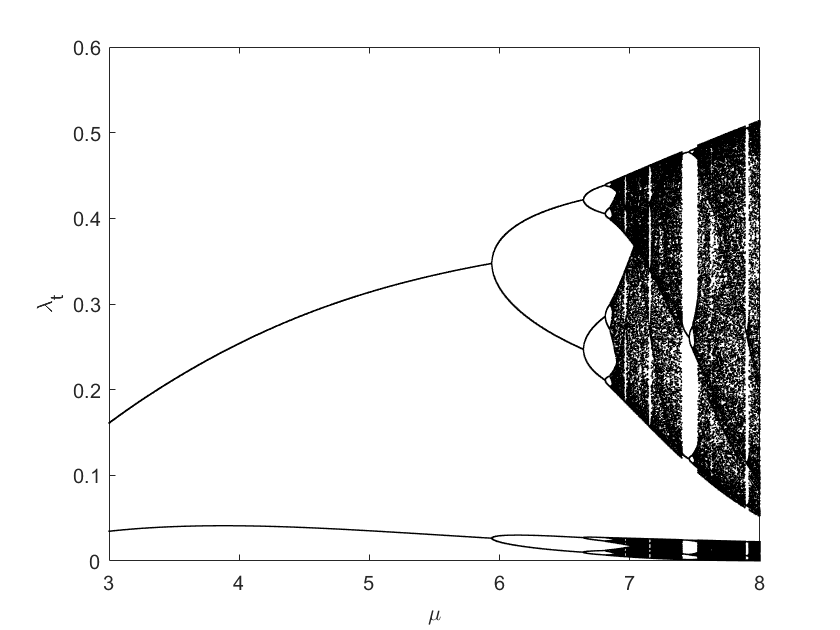}
    \caption{}
    \label{subfig:mu_1}
  \end{subfigure}
  \hfill
  \begin{subfigure}[b]{0.48\textwidth}
    \centering
    \includegraphics[width=\textwidth]{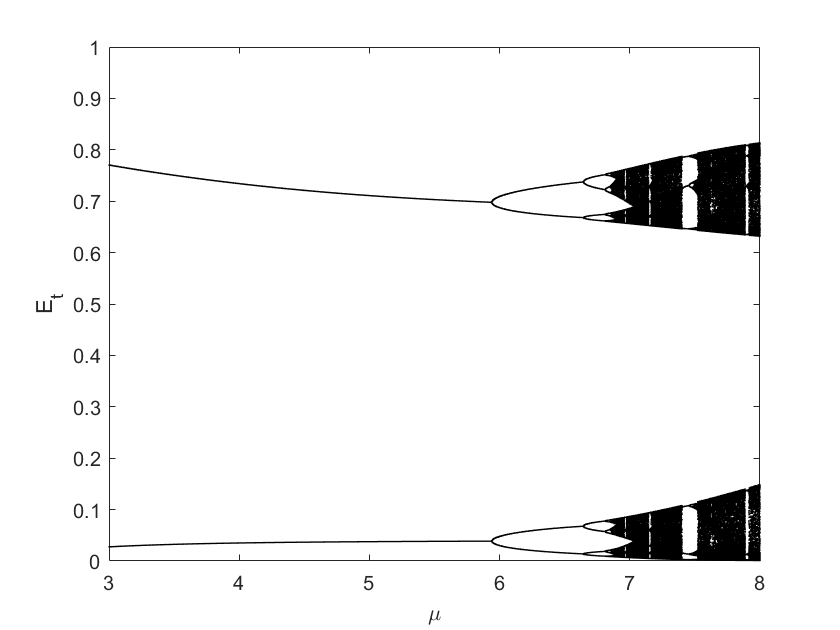}
    \caption{}
    \label{subfig:mu_2}
  \end{subfigure}
  \caption{Bifurcation Diagrams of behavioural reactivity. Parameter set: $\rho = 0.98$, $I = 1.0$, $p_e = 1.2$, $p_c = 0.53$, $\kappa = 0.3$, $\overline{\Pi}= 100$. }
  \label{fig:mu}
\end{figure}
\FloatBarrier
  
\section{Conclusion}
\label{sec:conclusion}

In this article, we have analysed the emergence of endogenous fluctuations and complex phenomena in a discrete-time dynamic model of educational choice. Our framework models a heterogeneous population composed of Followers (driven by imitative, bandwagon behaviour) and Positional Agents (driven by counter-adaptive, snobbish behaviour). Agents' preferences for education evolve endogenously, reacting not only to past aggregate enrolment ($E_t$) but also to an endogenous wage premium ($\Pi_t$) that declines with the supply of educated workers.

We have shown that in a polymorphic population, the interaction between these two behavioural types is a powerful source of instability. While the imitative drive of Followers ($\rho$) can amplify oscillations, we find that it is the counter-adaptive reactivity of Positional Agents (their sensitivity to status, $\sigma_E$, and economic returns, $\sigma_{\Pi}$) that acts as the primary destabilizing force. Our numerical exercises have highlighted the occurrence of complex educational cycles and the onset of chaotic regimes through a period-doubling (or flip) bifurcation cascade.

Furthermore, we have analysed the crucial role of the population structure itself. Rather than finding a simple stabilizing or destabilizing effect, we show that the population share of Followers, $\lambda$, has a non-monotonic impact on stability. Complex dynamics and chaos emerge only for intermediate, heterogeneous population mixes. Homogeneous populations, dominated by either pure Followers ($\lambda \to 1$) or pure Positional Agents ($\lambda \to 0$), tend to converge to simple, stable equilibria.

This finding suggests that the most complex collective outcomes arise precisely from the social conflict between the conforming drive of imitation and the counter-adaptive drive for distinction. From a policy standpoint, such endogenous volatility is a highly undesirable outcome. It hinders rational long-term planning for students, educational institutions, and the labour market. Indeed, in terms of policy, universities and public bodies cannot efficiently plan for future resource allocation. Periods of unexpectedly high demand would lead to overcrowded classrooms, diluted instructional quality, and strained budgets. These would inevitably be followed by periods of unexpectedly low demand, resulting in underutilized infrastructure, inefficient public spending, and faculty instability.

In the labour market, this volatility in the input ($E_t$) would translate directly into "boom-bust" cycles in the output (the supply of educated labour). The wage premium ($\Pi_t$), which is endogenous in our model, would itself become highly unstable. This scrambles the economic signals that education is meant to provide, making the returns on investment a gamble. Also, for individuals and households, this dynamic creates profound uncertainty for students and families. It becomes nearly impossible to make rational, long-term decisions about investing in human capital when the expected economic and social value of that investment is subject to wild, unpredictable swings.

Our analysis thus bridges the gap between non-linear social interaction models and the economics of education, highlighting the potential for intrinsic instability in systems driven by the interplay of economic incentives and peer pressure.

\subsection*{Acknowledgments}
No financial support was provided for this research.

\subsection*{Declaration of Interests}
The authors declare that they have no known competing financial interests or personal relationships that could have appeared to influence the work reported in this paper.


\subsection*{Declaration of generative AI and AI-assisted technologies in the manuscript preparation process}

During the preparation of this work the authors used ChatGPT, Gemini Pro, DeepSeek in order to code and proofread the manuscript. After using these tools, the authors reviewed and edited the content as needed and take full responsibility for the content of the published article.

\clearpage
\bibliography{reference}



\clearpage
\begin{appendices}

\section{Local derivations and proofs}
\label{app:local-derivations}

\subsection{Primitive objects and basic derivatives}

We recall
\[
C(E)=\frac{I}{p_c}-\frac{p_e}{p_c}E,\qquad 
\Pi(E)=\Pi-\kappa E,
\]
and, in the interior regime $\Pi>0$,
\[
C_E=-\frac{p_e}{p_c},\qquad \Pi_E=-\kappa.
\]
Preference parameters (time indices omitted for readability):
\[
\begin{aligned}
\alpha_{F}(E) &= 1-\exp\!\big(-\rho E-\rho_\pi \Pi(E)\big), 
&\qquad \beta_{F}(E) &= 1-\exp\!\big(-\rho C(E)\big),\\
\alpha_{P}(E) &= \exp(-\sigma E)\,\Big(1-\exp\big(-\sigma_\Pi \Pi(E)\big)\Big), 
&\qquad \beta_{P}(E) &= \exp\!\big(-\sigma C(E)\big).
\end{aligned}
\]
Define the shares
\[
s_i(E):=\frac{\alpha_i(E)}{\alpha_i(E)+\beta_i(E)}\in(0,1),\qquad i\in\{F,P\}.
\]

Derivatives of $\alpha_i$ and $\beta_i$.\\
Set $A:=-\rho E-\rho_\pi \Pi$, $B:=-\rho C$, $\widetilde A:=-\sigma E$, $\widetilde B:=-\sigma_\Pi \Pi$, $\widetilde C:=-\sigma C$. Using $C_E$ and $\Pi_E$:
\[
\begin{aligned}
\alpha_{F,E} &= \frac{d}{dE}\Big(1-e^{A}\Big) = -e^{A}\,A_E
=(\rho-\rho_\pi\kappa)\,e^{A},\\[0.2em]
\beta_{F,E} &= \frac{d}{dE}\Big(1-e^{B}\Big)= -e^{B}\,B_E
=\frac{\rho p_e}{p_c}\,e^{B},\\[0.3em]
\alpha_{P,E} &= \frac{d}{dE}\Big(e^{\widetilde A}\,(1-e^{\widetilde B})\Big)
= e^{\widetilde A}\Big[-\sigma(1-e^{\widetilde B}) + e^{\widetilde B}\,\sigma_\Pi\kappa\Big]\\
&= -\sigma e^{\widetilde A}\!(1-e^{\widetilde B}) - \kappa\sigma_\Pi e^{\widetilde A+\widetilde B},\\[0.2em]
\beta_{P,E} &= \frac{d}{dE}\big(e^{\widetilde C}\big)= e^{\widetilde C}\,\widetilde C_E
=\frac{\sigma p_e}{p_c}\,e^{\widetilde C}.
\end{aligned}
\]

Derivative of $s_i(E)$.\\
With $s_i=\alpha_i/(\alpha_i+\beta_i)$ and quotient rule:
\begin{equation}
\label{eq:app-siE}
s_{i,E}(E)=\frac{\alpha_{i,E}\,\beta_i-\alpha_i\,\beta_{i,E}}{(\alpha_i+\beta_i)^2},\qquad i\in\{F,P\}.
\end{equation}

\subsection{Individual demands and indirect utility}

Given $U_{i}(e,c;\alpha_i,\beta_i)=\alpha_i\log e+\beta_i\log c$ with prices $(p_e,p_c)$ and income $I$, the optimal demands (Cobb–Douglas) are
\[
e_i^\ast=\frac{\alpha_i}{\alpha_i+\beta_i}\,\frac{I}{p_e}=\; s_i\,\frac{I}{p_e},\qquad
c_i^\ast=\frac{\beta_i}{\alpha_i+\beta_i}\,\frac{I}{p_c}=\; (1-s_i)\,\frac{I}{p_c}.
\]
The indirect utility is
\begin{align}
\mathcal U_i(E)
&= \alpha_i\,\log\!\Big(s_i\frac{I}{p_e}\Big)+\beta_i\,\log\!\Big((1-s_i)\frac{I}{p_c}\Big)\nonumber\\
&= \alpha_i\Big[\log(I/p_e)+\log s_i\Big]+\beta_i\Big[\log(I/p_c)+\log(1-s_i)\Big].
\label{eq:app-Ui}
\end{align}

Derivative $\mathcal U_{i,E}$.\\
Differentiate \eqref{eq:app-Ui} using product rule and $s_i=\alpha_i/(\alpha_i+\beta_i)$:
\[
\frac{d}{dE}\big(\alpha_i\log s_i\big)=\alpha_{i,E}\log s_i+\alpha_i\,\frac{s_{i,E}}{s_i},\quad
\frac{d}{dE}\big(\beta_i\log(1-s_i)\big)=\beta_{i,E}\log(1-s_i)-\beta_i\,\frac{s_{i,E}}{1-s_i}.
\]
But $\displaystyle \alpha_i\frac{1}{s_i}-\beta_i\frac{1}{1-s_i}
=\frac{\alpha_i+\beta_i}{s_i(1-s_i)}\Big(s_i-(1-s_i)\Big)=0$ because $s_i=\alpha_i/(\alpha_i+\beta_i)$. Hence the $s_{i,E}$-terms cancel out and we obtain the compact formula
\begin{equation}
\label{eq:app-UiE-compact}
\quad 
\mathcal U_{i,E}
=\alpha_{i,E}\!\Big[\log(I/p_e)-\log s_i\Big]
+\beta_{i,E}\!\Big[\log(I/p_c)-\log(1-s_i)\Big],
\quad i\in\{F,P\}.
\end{equation}
Equivalently, expanding logs and re-grouping yields the alternative expression used in the main text:
\begin{align}
\mathcal U_{i,E}
&=\alpha_{i,E}\left[\log\!\left(\frac{I}{p_e}\right)+\log\!\left(\frac{\alpha_i}{\alpha_i+\beta_i}\right)+1-\frac{\alpha_i}{\alpha_i+\beta_i}\right]\nonumber\\
&\quad+\beta_{i,E}\left[\log\!\left(\frac{I}{p_c}\right)+\log\!\left(\frac{\beta_i}{\alpha_i+\beta_i}\right)+1-\frac{\beta_i}{\alpha_i+\beta_i}\right],
\label{eq:app-UiE-expanded}
\end{align}
which is algebraically equivalent to \eqref{eq:app-UiE-compact}.

\subsection{Aggregate map and composition dynamics}

The aggregate map and the composition rule are
\[
\Gamma(E;\lambda)=\frac{I}{p_e}\Big[\lambda\,s_F(E)+(1-\lambda)\,s_P(E)\Big],\qquad
\mathcal V(E,\lambda)=\frac{e^{\mu\mathcal U_F(E)}}{e^{\mu\mathcal U_F(E)}+e^{\mu\mathcal U_P(E)}}.
\]

Partials of $\Gamma$.\\
Using \cref{eq:app-siE}:
\begin{equation}
\label{eq:app-Gamma-partials}
\Gamma_E(E;\lambda)=\frac{I}{p_e}\Big[\lambda\,s_{F,E}(E)+(1-\lambda)\,s_{P,E}(E)\Big],\qquad
\Gamma_\lambda(E;\lambda)=\frac{I}{p_e}\Big[s_F(E)-s_P(E)\Big].
\end{equation}

Partials of $\mathcal V$.\\
For the logit in indirect utilities, differentiate w.r.t.\ $E$ and $\lambda$:
\begin{equation}
\label{eq:app-V-partials}
\mathcal V_\lambda(E,\lambda)=0,\qquad
\mathcal V_E(E,\lambda)=\mu\,\mathcal V(E,\lambda)\big(1-\mathcal V(E,\lambda)\big)\,\Big(\mathcal U_{F,E}(E)-\mathcal U_{P,E}(E)\Big).
\end{equation}
The bound $0\le \mathcal V(1-\mathcal V)\le \tfrac14$ holds by concavity of $x(1-x)$ on $[0,1]$.

\subsection{Jacobian, Schur conditions, and bifurcation curves}

At an interior fixed point $(E^\star,\lambda^\star)$, the Jacobian of $\Phi(E,\lambda)=(\Gamma,\mathcal V)$ is
\begin{equation}
\label{eq:app-J}
J(E^\star,\lambda^\star)=
\begin{pmatrix}
\Gamma_E & \Gamma_\lambda\\
\mathcal V_E & 0
\end{pmatrix}_{(E^\star,\lambda^\star)},\qquad
\tau:=\mathrm{tr}\,J=\Gamma_E,\quad
\Delta:=\det J=-\,\Gamma_\lambda\,\mathcal V_E.
\end{equation}
By the Schur–Cohn criterion, the eigenvalues lie in the open unit disk iff
\begin{equation}
\label{eq:app-Schur}
1-\tau+\Delta>0,\qquad 1+\tau+\Delta>0,\qquad 1-\Delta>0,
\end{equation}
i.e.,
\[
1-\Gamma_E-\Gamma_\lambda\mathcal V_E>0,\qquad
1+\Gamma_E-\Gamma_\lambda\mathcal V_E>0,\qquad
1+\Gamma_\lambda\mathcal V_E>0.
\]
Local bifurcations occur on the codimension-one curves:
\begin{align}
&\text{saddle-node:} && 1-\Gamma_E-\Gamma_\lambda\mathcal V_E=0,\label{eq:app-SN}\\
&\text{flip (period-2):} && 1+\Gamma_E-\Gamma_\lambda\mathcal V_E=0,\label{eq:app-flip}\\
&\text{Neimark--Sacker:} && -\,\Gamma_\lambda\mathcal V_E=1,\quad -2<\Gamma_E<2.\label{eq:app-NS}
\end{align}

\subsection{A sufficient $\ell_1$-type bound and a threshold on $\mu$}
\label{app:mu-threshold-proof}

Let $x:=\Gamma_E$ and $y:=-\Gamma_\lambda\mathcal V_E=\Delta$. If
\begin{equation}
\label{eq:app-l1}
|x|+|y|<1,
\end{equation}
then $1-x+y\ge 1-|x|-|y|>0$, $1+x+y\ge 1-|x|-|y|>0$, and $1-y\ge 1-|y|>0$. Hence \eqref{eq:app-Schur} holds and the fixed point is locally asymptotically stable.

Moreover, since $\Gamma_\lambda=\tfrac{I}{p_e}(s_F-s_P)$ and $0\le \mathcal V(1-\mathcal V)\le \tfrac14$,
\begin{equation}
\label{eq:app-GLV-bound}
|\Gamma_\lambda\,\mathcal V_E|
=\left|\frac{I}{p_e}(s_F-s_P)\,\mu\,\mathcal V(1-\mathcal V)\,\big(\mathcal U_{F,E}-\mathcal U_{P,E}\big)\right|
\le \frac{I}{p_e}\cdot\frac{\mu}{4}\cdot \big|\mathcal U_{F,E}-\mathcal U_{P,E}\big|.
\end{equation}
Setting
\[
g^\star:=|\Gamma_E(E^\star,\lambda^\star)|,\qquad
h^\star:=|\mathcal U_{F,E}(E^\star)-\mathcal U_{P,E}(E^\star)|,
\]
a sufficient condition for \eqref{eq:app-l1} is
\[
g^\star+\frac{I}{p_e}\cdot\frac{\mu}{4}\,h^\star<1
\quad\Longleftrightarrow\quad
\mu<\bar\mu:=\frac{4p_e}{I}\,\frac{1-g^\star}{h^\star},
\]
with $0\le g^\star<1$. This yields Proposition~\ref{prop:mu-threshold} in the main text.

\subsection{Conservative computable bounds for $g^\star$ and $h^\star$}

Using $s_i\in(0,1)$ and \eqref{eq:app-siE}:
\[
|\Gamma_E|\le \frac{I}{p_e}\Big(\lambda^\star|s_{F,E}|+(1-\lambda^\star)|s_{P,E}|\Big)
\le \frac{I}{p_e}\max\{|s_{F,E}|,|s_{P,E}|\}.
\]
With $\alpha_i,\beta_i\in(0,1)$ and the closed-form derivatives above, one can build an explicit upper bound $\widehat g^\star\ge g^\star$. Similarly, from \eqref{eq:app-UiE-compact},
\[
|\mathcal U_{i,E}|\le |\alpha_{i,E}|\Big|\log(I/p_e)-\log s_i\Big|
+|\beta_{i,E}|\Big|\log(I/p_c)-\log(1-s_i)\Big|,
\]
so that a conservative $\widehat h^\star\ge h^\star$ follows by bounding $s_i$ away from $0$ and $1$ using the calibrated parameter ranges. Substituting $(\widehat g^\star,\widehat h^\star)$ in $\bar\mu$ gives a sufficient (possibly tighter) threshold.

\end{appendices}

\end{document}